\documentclass[11pt]{article}
\usepackage{cite}
\usepackage{graphicx}
\usepackage{url}
\usepackage{amsmath}
\usepackage{amsthm}
\newtheorem{thm}{Theorem}[section]
\newtheorem{lem}[thm]{Lemma}
\newtheorem{cor}[thm]{Corollary}

\begin{document}

\title{Notes on Reed-Muller codes}

\author{{\normalsize Yanling Chen} \\
        {\small Q2S, Centre of Excellence}\\
        {\small Norwegian University of Science and Technology}\\
        {\small Trondheim, Norway}\\ 
        \quad\\    
        {\normalsize Han Vinck} \\
        {\small Institute for Experimental Mathematics}\\
        {\small University of Duisburg-Essen}\\
        {\small Essen 45326, Germany}\\       
        {\small Email: \{julia, vinck\}@iem.uni-due.de}}

\date{}
\maketitle

\begin{abstract} In this paper, we consider the Reed-Muller (RM) codes.
For the first order RM code, we prove that it is unique in the sense
that any linear code with the same length, dimension and minimum
distance must be the first order RM code; For the second order RM code,
we give a constructive {\it linear} sub-code family for the case
when m is even.
This is an extension of Corollary 17 of Ch. 15 in the coding book by
MacWilliams and Sloane. Furthermore, we show that the specified
sub-codes of length $\leq 256$ have minimum distance equal to the
upper bound or the best known lower bound for all linear codes of
the same length and dimension.
As another interesting result, we derive an additive commutative
group of the symplectic matrices with full rank.
\end{abstract}


\section{Introduction}
Let $C$ be an $[n,k,d_{min}]$ binary linear code of length $n$,
dimension $k$ and minimum distance $d_{min}.$ Let $V=\{0,1\}$ and
let $v=(v_1,\cdots, v_m)$ range over $V^m,$ the set of all binary
$m$-tuples. Any function $f(v)=f(v_1,\cdots,v_m)$ which takes on
the values 0 and 1 is called a {\it Boolean function}. Reed-Muller
(or $\mathrm{RM}$) codes can be defined very simply in terms of
Boolean functions. As stated in \cite{6}, the $r$th order binary
$\mathrm{RM}$ code $\mathcal{R}(r,m)$ of length $n=2^m,$ for
$0\leq r\leq m,$ is the set of all vectors $f,$ where
$f(v_1,\cdots,v_m)$ is a Boolean function of degree at most $r.$
In this paper, we consider the RM codes $\mathcal{R}(r,m)$ for
$r=1,2.$ The standard references are Ch. 14 and Ch. 15 in
\cite{6}. Necessary preliminaries are introduced at the beginning
of Section \ref{sec: 1st order RM} and Section \ref{sec: 2nd order
RM}, respectively.

The first order $\mathrm{RM}$ code is a linear $[2^m, 1+m,
2^{m-1}]$ code. It is optimal in the sense that it reaches the
Plotkin bound (refer to Theorem 8 of Ch. 2 in \cite{6}). Besides,
it has a simple weight distribution with one all-zero codeword,
one all-one codeword and $2(2^{m}-1)$ codewords of weight
$2^{m-1}.$ In Section \ref{sec: 1st order RM}, we show that it has
another interesting property: uniqueness, in the sense that, any linear code
with the same length, dimension and minimum distance must be the
first order $\mathrm{RM}$ code.

For the second order $\mathrm{RM}$ code, we recall that there is a
linear sub-code family  for odd $m$ as given in Theorem \ref{thm:
m odd sub-code family} (also Corollary 17 of Ch. 15 in \cite{6}).
Note that the sub-code $\mathcal{R}_{1,2t+1}^{t}$ (or
$\mathcal{R}_{2,2t+1}^{t}$) is optimal in the sense that no codes
exist with the same length and weight set, but with larger
dimension (refer to Proposition 14 in \cite{7}). Correspondingly
for even $m,$ a nonlinear sub-code family (the generalized Kerdock
code $\mathcal{DG}(m,d)$) is introduced in Theorem 19. of Ch. 15
in \cite{6}. In particular, when $d=m/2,$ $\mathcal{DG}(m,d)$ is
one of the best-known nonlinear codes: the Kerdock code
$\mathcal{K}(m).$

As shown in \cite{6}, these nonlinear sub-codes turn out to have
good parameters. However, a linear sub-code family for even $m$ is
still of our interest due to the advantages of linear codes in the
straightforward decoding and practical implementations. In Section
\ref{sec: 2nd order RM}, we reconsider the case for even $m$ and
build up a linear sub-code family as given in Theorem \ref{thm: m
even sub-code family}. At the end, we conclude in Section
\ref{sec: conclusion}.

\begin{thm} \label{thm: m odd sub-code family}
See Corollary 17 of Ch. 15 in \cite{6}. Let $m=2t+1$ be odd, and
let $d$ be any number in the range $1\leq d\leq t.$ Then there
exist two
    {
    \begin{equation*}
        [2^m, m(t-d+2)+1, 2^{m-1}-2^{m-d-1}]
    \end{equation*}
    }
sub-codes $\mathcal{R}_{1, 2t+1}^{d}$ and $\mathcal{R}_{2,
2t+1}^{d}$ of $\mathcal{R}(2, m).$ These are obtained
respectively, by extending the cyclic sub-codes of
$\mathcal{R}(2,m)^*$ having idempotents
    {
    \begin{equation*}
        \theta_0+\theta_1^*+\sum_{j=d}^{t}\theta_{l_j}^* \quad
        \mbox{and} \quad
        \theta_0+\theta_1^*+\sum_{j=1}^{t-d+1}\theta_{l_j}^*,
        \quad l_j=1+2^j.
    \end{equation*}
    }
These codes have weights $2^{m-1}$ and $2^{m-1}\pm 2^{m-h-1}$ for
all $h$ in the range $d\leq h\leq t.$
\end{thm}

\section{The first order Reed-Muller code} \label{sec: 1st order RM}
The first order Reed-Muller code $\mathcal{R}(1,m)$ consists of
all vectors $u_0 \mathbf{1}+\sum_{i=1}^{m} u_i \mathbf{v_i},$
$u_i=0 \mbox{ or } 1,$ corresponding to the linear Boolean functions.
Define the {\it orthogonal code} $\mathcal{O}_m$ to be the
$[2^m,m,2^{m-1}]$ code consisting of the vectors $\sum_{i=1}^{m}
u_i \mathbf{v_i}.$  Then
                        {
                        \begin{equation*}
                           \mathcal{R}(1,m)=\mathcal{O}_m\cup(\mathbf{1}+\mathcal{O}_m),
                        \end{equation*}
                        }
is a $[2^m,1+m,2^{m-1}]$ code.

\begin{thm}\label{thm: unique 1rm and simplex}
(Uniqueness of the first order RM code) Any linear code with
parameters $[2^m, 1+m, 2^{m-1}]$ is equivalent to the first order
Reed-Muller code. They have a unique weight distribution
$1+2(2^m-1)t^{2^{m-1}}+t^{2^m}.$
\end{thm}

\begin{proof} Let $C$ be any linear $[2^m, 1+m, 2^{m-1}]$ code. Let $G$ be its
generator matrix of systematic form. We show in the following 4
steps that $C$ must be $\mathcal{R}(1,m).$

            {\it 1st step: Construct code $C_1$ from $C$.}

            Deleting the first row and the first column of $G,$
            we get a new generator matrix $G_1$. Let $C_1$ be the code
            generated by $G_1.$ Then $C_1$ is a linear $[2^m-1, m, d_1]$
            code, where $d_1 \geq 2^{m-1}.$

            Let $n_1=2^m-1, k_1=m$ and assume that $d_1\geq 2^{m-1}+1.$ According to the Plotkin
            bound, the number of the codewords of $C_1$ must satisfy the following inequality
                {
                \begin{equation*}
                    |C_1|\leq 2[\frac{d_1}{2d_1-n_1}]\leq 2[d/3]<d\leq n_1=2^m-1<2^m.
                \end{equation*}
                }
             This contradicts the fact $|C_1|=2^{k_1}=2^{m}.$ So $d_1=2^{m-1}.$

            {\it 2nd step: $C_1$ is the Simplex code.}

            We consider the dual code of $C_1,$ $C_1^{\bot},$ a
            linear  $[2^m-1, 2^m-m-1, d_1^{\bot}]$ code, where $d_1^{\bot}\geq 1.$

            Suppose $d_1^{\bot}=1.$ Then there is a codeword $\mathbf{a}$ of weight 1 in
            $C_1^{\bot}$ with only one 1-entry, whose
            coordinate is assumed to be $i.$ Since $\mathbf{a}\in C_1^{
            \bot},$ correspondingly the $i$-th position of all
            codewords in $C_1$ must be 0. Deleting the $i$-th
            coordinate of code $C_1,$ we get a linear $[2^m-2, m, 2^{m-1}]$
            code, which is impossible due to the Plotkin bound.

            Suppose $d_1^{\bot}=2.$ Then there is a codeword $\mathbf{b}$ of weight 2 with only two
            1-entries, whose coordinates are assumed to be $i$ and $j.$
            Since $\mathbf{b}\in C_1^{
            \bot},$ correspondingly the $i$-th position and the $j$-th position of all
            codewords in $C_1$ must have the same entries. Note that it is impossible that
            the $i$-th and $j$-th positions of all codewords have
            only 1-entry. Suppose that they are all 0-entry. Then by deleting the
            $i$-th and $j$-th coordinates of code $C_1,$ we can get
            a linear $[2^m-3, m, 2^{m-1}]$ code.
            However, such a linear code is not existent due to the Plotkin bound.
            Now we consider the case when the $i$-th and $j$-th positions of all
            codewords in $C_1$ have not only 0-entry but also
            1-entry. We take the codewords with 0-entry at
            positions $i$ and $j$ and then delete both
            coordinates. The derived codewords build up a linear $[2^m-3, m-1,
            2^{m-1}]$ code, which is also not existent due to the
            Plotkin bound.

            As a conclusion of above discussion, we have $d_1^{\bot}\geq 3.$ It is easy to verify
            that when $e=\lfloor \frac{d_1^{\bot}-1}{2}\rfloor=1,$ $C_1^{\bot}$
            reaches the Hamming bound (refer to Theorem 6 of Ch. 1 in \cite{6})
             and thus $C_1^{\bot}$ is a perfect single-error-correcting code.
            Since the binary linear perfect single-error-correcting code with the same length and dimension is unique
            \cite{12}, so $d_1^{\bot}=3$ and $C_1^{\bot}$ must be the binary Hamming code.
            Therefore, $C_1$ is the Simplex code and has a
            unique weight enumerator $1+(2^m-1)t^{2^{m-1}}.$
            So far we have proved that any binary linear with parameters $[2^m-1, m, 2^{m-1}]$
            must be the Simplex code.

            {\it 3rd step: Construct code $C$ back from code $C_1.$}

            Recall that $G$ is a generator matrix of $C$ in a
            systematic form. We get $G_1,$ the generator matric of
            $C_1,$ by deleting the first row and the first column
            of $G.$ Now in order to get back $C$ from $C_1,$ first we
            add an all-zero column to $G_1$ and denote the extended
            matrix as $G_2.$ Then we add a row to $G_2$ so as to get a matrix equivalent to
            $G$ and thus code $C.$

            Since $C_1$ is the Simplex code, so the columns of
            $G_1$ are the binary representations of the numbers from 1 to
            $2^m-1.$ Note that $G_2$ is constructed by adding an
            all-0 column into $G_1.$ So the columns of $G_2$ go
            through all binary representations of the numbers from 0
            to $2^m-1.$ Clearly $C_2$ is equivalent to the orthogonal code
            $\mathcal{O}_m.$ Without loss of generality, we use the Boolean
            functions ${v}_1, \cdots, {v}_m$ to denote the basis rows of $G_2.$
            Then the basis codeword $\mathbf{c}$ added into $G_2$ can be
            specified as a Boolean function $f({v})=f({v}_1, \cdots,
            {v}_m).$ Clearly $C=C_2\cup(\mathbf{c}+C_2)=\mathcal{O}_m\cup(f+\mathcal{O}_m).$

            {\it 4th step: Code $C$ is the first order Reed-Muller code.}

            For any vector $u=(u_1, \cdots, u_m)\in
            V^m,$ $f(u)$ denote the value of the Boolean function $f$ at $u.$
            Let $F(u)=(-1)^{f(u)}.$ The Hadamard transform of $F$ (refer to p. 414 of Ch. 14 in \cite{6}) is
            given by
                    {
                    \begin{eqnarray*}
                        \hat{F}(u)&=&\sum_{v\in V^m} (-1)^{u\cdot v}
                        F(v), \quad u\in V^m, \\
                        &=& \sum_{v\in V^m} (-1)^{u\cdot v+f(v)}.
                    \end{eqnarray*}
                    }
            $\hat{F}(u)$ is equal to the number of 0's minus the
            number of 1's in the binary vector  $f+\sum_{i=1}^m u_i {v}_i.$
            Thus
                    {
                        \begin{equation*}
                            \hat{F}(u)=2^m-2\mbox{dist}\{f, \sum_{i=1}^m u_i {v}_i\}.
                        \end{equation*}
                    }
            Here $\mbox{dist}\{\mathbf{a}, \mathbf{b}\}$ is the
            Hamming distance of two binary vectors $\mathbf{a}$ and $\mathbf{b}.$
            Note that $f+\sum_{i=1}^m u_i {v}_i\in C$ and
            code $C$ has minimum distance $2^{m-1},$ i.e.,
                    {
                        \begin{equation*}
                           \mbox{dist}\{f, \sum_{i=1}^m u_i {v}_i\}\geq 2^{m-1}.
                        \end{equation*}
                    }
            Therefore,
                    {
                        \begin{equation*}
                             \hat{F}(u)\leq 0,\ \mbox{for any}\ u\in V^m.
                        \end{equation*}
                    }
            According to Lemma 2 of Ch. 14 in \cite{6}, for any $v\neq
            0,$ $v\in V^m,$
                    {
                        \begin{equation*}
                                \sum_{u\in V^{m}} \hat{F}(u)\hat{F}(u+v)=0.
                        \end{equation*}
                        }
            Let $W$ be the set such that
            $\hat{F}(u)=0$ for $u\in W.$ Then for $u\in V^m\setminus W,$ we
            have $\hat{F}(u)<0.$ For any $v\neq 0,$ $v\in
            V^m,$
                    {
                        \begin{align*}
                                \sum_{u\in V^{m}} \hat{F}(u)\hat{F}(u+v)
                                    =&\sum_{u\in \{W\cup(-v+W)\}} \hat{F}(u)\hat{F}(u+v)
                                    + \sum_{u\in V^{m}\setminus\{W\cup(-v+W)\}} \hat{F}(u)\hat{F}(u+v)\\
                                    =&\sum_{u\in V^{m}\setminus\{W\cup(-v+W)\}} \hat{F}(u)\hat{F}(u+v)\\
                                    =&0,
                        \end{align*}
                    }
            where $-v+W=\{u-v| u\in W\}.$ Clearly, we have
                    {
                        \begin{equation*}
                              V^{m}=W\cup(-v+W), \ \mbox{for any}\ v\in V^m\setminus\{0\}.
                        \end{equation*}
                    }

            In the following we will prove that $|W|=2^m-1.$
            First we show that $|W|\neq 2^m,$ i.e., $W\neq
            V^m.$ Suppose that $W=V^m,$ then we have
                    {
                        \begin{equation*}
                            \sum_{u\in V^{m}} \hat{F}(u)^2=\sum_{u\in W}
                            \hat{F}(u)^2=0.
                        \end{equation*}
                        }
            This is a  contradiction to the Parseval's equation (refer to Corollary 3 of Ch. 14 in \cite{6})
                    {
                        \begin{equation*}
                            \sum_{u\in V^{m}} \hat{F}(u)^2=2^{2m}.
                        \end{equation*}
                    }

            Now we prove that $|W|>2^m-2.$ Suppose that $|W|\leq
            2^m-2.$ Then there are at least two elements $i, j\in V^m,$ $i\neq j$ and
             $i,j\notin W.$  Note that if we choose $v=j-i\neq 0,$ then  we have $i\notin -v+W$ since
            $j\notin W.$
            Thus $i\in V^m$ but $i\notin \{W\cup (-v+W)\}$ for $v=j-i\neq 0.$ This contradicts
            to the fact that
                    {
                        \begin{equation*}
                                V^{m}=W\cup(-v+W), \ \mbox{for any}\ v\in V^m\setminus\{0\}.
                        \end{equation*}
                    }
            As a result, we can conclude that $|W|=2^m-1.$ Assume
            that $\tilde{u}\notin W.$ Recall the Parseval's equation.
            We have
                    {
                        \begin{equation*}
                            \sum_{u\in V^{m}} \hat{F}(u)^2=\hat{F}(\widetilde{u})^2=2^{2m}.
                        \end{equation*}
                    }
            Since $\hat{F}(\widetilde{u})<0,$ it is clear that $\hat{F}(\widetilde{u})=-2^m.$
            Note that $\hat{F}(u)=2^m-2\mbox{dist}\{f, \sum_{i=1}^m u_i {v}_i\}.$
            So
                    {
                        \begin{equation*}
                            \mbox{dist}\{f, \sum_{i=1}^m u_i
                            {v}_i\}=\frac{1}{2}\{2^m-\hat{F}(u)\}
                            =\left\{
                                    \begin{tabular}{ll}
                                        $2^{m-1}$ & $u\in W,$ \\
                                        $2^m$     & $u=\widetilde{u}.$
                                    \end{tabular}
                            \right.
                        \end{equation*}
                    }
            In other words, the coset $\mathbf{c}+C_2$ or
            $f+\mathcal{O}_m$ has one codeword of weight $2^m$ and
            $2^m-1$ codewords of weight $2^{m-1}.$ Thus, we have
            proved that $\mathbf{1}\in
            \mathbf{c}+C_2.$ So
                    {
                        \begin{equation*}
                             C=C_2\cup(\mathbf{1}+C_2)=\mathcal{O}_m\cup(\mathbf{1}+\mathcal{O}_m),
                        \end{equation*}
                        }
            is the first order Reed-Muller code and has the unique weight
            enumerator $1+(2^{m+1}-2)t^{2^{m-1}}+t^{2^m}.$
\end{proof}

\section{The second order Reed-Muller code}\label{sec: 2nd order RM}

The second order binary $\mathrm{RM}$ code $\mathcal{R}(2,m)$ is the
set of all vectors $f,$ where $f(v_1,\cdots,v_m)$ is a Boolean
function of degree $\leq 2.$ $\mathcal{R}(2,m)$ is of length
$2^m,$ of dimension $1+m+\binom{m}{2}$ and of minimum distance
$2^{m-2}.$ A typical codeword of $\mathcal{R}(2,m)$ is given by
the Boolean function
    {
    \begin{eqnarray*}
        S(v)&=&\sum_{1\leq i\leq j\leq m}q_{ij} v_i v_j +\sum_{1\leq i\leq m}u_i v_i +\epsilon\\
            &=&vQv^{T}+Lv^{T}+\epsilon,
    \end{eqnarray*}
    }
where $v=(v_1,\cdots,v_m),$
$Q=(q_{ij})$ is an upper triangular binary matrix,
$L=(u_1,\cdots,u_m)$ is a binary vector and $\epsilon$ is 0 or 1.
Note that, if $Q$ is fixed and the linear function
$Lv^{T}+\epsilon$ varies over the first-order RM code
$\mathcal{R}(1,m),$ then $S(v)$ runs through a coset of
$\mathcal{R}(1,m)$ in $\mathcal{R}(2,m).$ This coset is
characterized by $Q,$ or alternatively by the symmetric matrix
$B=Q+Q^{T}.$ $B$ is a binary symmetric matrix with zero diagonal,
which is called {\it symplectic matrix}. The {\it symplectic form}
associated with $B,$ as defined by (4) of Ch. 15 in \cite{6}, is
    {
    \begin{eqnarray*}
        \mathcal{B}(u,v)&=&u(Q+Q^{\mathrm{T}})v^{\mathrm{T}}\\
                        &=& S(u+v)+S(u)+S(v)+\epsilon.
    \end{eqnarray*}
    }
As stated in Lemma \ref{lem: rank dis}, the weight distribution of
the coset associated with $\mathcal{B}$ depends only on the rank
of the matrix $B.$
\begin{lem}\label{lem: rank dis}
See Theorem 5 of Ch. 15 in \cite{6}. If the matrix $B$ has rank
$2h,$ the weight distribution of the corresponding coset of
$\mathcal{R}(1,m)$ in $\mathcal{R}(2,m)$ is as follows:
    {
        \begin{equation*}
            \begin{tabular}{cc}
                \hline
                \mbox{Weight}       & \mbox{Number of Vectors}\\
                \hline
                $2^{m-1}-2^{m-h-1}$ & $2^{2h}$\\
                $2^{m-1}$           & $2^{m+1}-2^{2h+1}$\\
                $2^{m-1}+2^{m-h-1}$ & $2^{2h}$\\
                \hline
            \end{tabular}
        \end{equation*}
        }
\end{lem}

Since rank of $B$ satisfies $2h\leq m$, the coset with largest possible minimum weight occurs when $2h=m.$ In this
case $m$ must be even. Furthermore, the Boolean functions
associated with such cosets are quadratic {\it bent functions}, in
the sense that they are furthest away from the linear Boolean
functions (refer to Theorem 6 of Ch. 14 in \cite{6}).

Note that the binary RM code is also conveniently defined as an
extension of a cyclic code. In this section, we consider the
punctured second order RM code $\mathcal{R}(2,m)^*.$ As shown in
Theorem \ref{thm: m even sub-code family}, we obtain a sub-code
family of $\mathcal{R}(2,m)$ for even $m$ by extending a family of
sub-codes of $\mathcal{R}(2,m)^*.$

Let $F=GF(2)$ and $F[x]$ be the set of polynomials in $x$ with
coefficients from $F.$ Define the ring $R_n=F[x]/(x^n-1),$ which
consists of the residue classes of $F[x]$ modulo $x^n-1.$ A cyclic
code of length $n$ is an ideal of $F[x]/(x^n-1).$ Let $GF(2^m)$ be
the splitting field of $x^n-1$ over $F$ ($m$ is the smallest
positive integer such that $n$ divides $2^m-1$). Let $\alpha\in
GF(2^m)$ be a primitive $n$-th root of unity. If
$c(x)=\sum_{i=0}^{n-1} c_i x^i$ is a polynomial of $R_n,$ then the
Mattson-Solomon polynomial of $c(x)$ is a polynomial in $F[z]$
defined by
    {
    \begin{equation*}
        A(z)=\sum_{j=1}^{n}c(\alpha^j)z^{n-j}.
    \end{equation*}
    }
It has the following property.
    \begin{thm} \label{thm: Mattson-Solomon polynomial coe} See Theorem 1.2 in
    \cite{13}. $c_i=A(\alpha^i).$
    \end{thm}

We partition the integers $\mathrm{mod}\ n$ into sets called {\it
cyclotomic cosets} $\mathrm{mod}\ n.$ The cyclotomic coset containing $s$
is
    {
    \begin{equation*}
        C_s=\{s,2s, 2^2s, 2^3s,\cdots, 2^{m_s-1}s\},
    \end{equation*}
    }
where $m_s$ is the smallest positive integer such that
$2^{m_s}\cdot s\equiv s(\mathrm{mod}\ n).$ Clearly, the numbers $-s,$ $
-2s,$ $-2^2s,$ $-2^3s,$ $\cdots,$ $-2^{m_s-1}s$ also form a
cyclotomic coset denoted $C_{-s}.$ Let
      {
    \begin{equation*}
        T_{m_s}(\gamma)=\gamma+\gamma^2+\gamma^{2^2}+\cdots+\gamma^{2^{m_s-1}}.
    \end{equation*}
    }
If $\gamma\in GF(2^{m_s}),$ then it is called the {\it trace} of
$\gamma$ from $GF(2^{m_s})$ to $GF(2).$

It is known that $R_n$ is the direct sum of its minimal ideals and
each minimal ideal is generated by a {\it primitive idempotent}.
A primitive idempotent $\theta_s$ is defined by the property
    {
    \begin{equation*}
       \theta_s(\alpha^j)=\left \{
                                \begin{array}{ll}
                                    1 & \mbox{if}\ j\in C_s,\\
                                    0 & \mbox{otherwise}.
                                \end{array}
                            \right.
    \end{equation*}
    }
Let $\theta_s^*(x)=\theta_{s'}(x),$ where $s'\in C_{-s}.$ The
polynomial $\theta_s^*(x)$ is also a primitive idempotent. Its
Mattson-Solomon polynomial is $A(z)=\sum_{j\in C_s} z^j.$

Let ${\varPsi}$ be a cyclic code of length $n$ over $F$. Then
there is a unique polynomial in ${\varPsi}$ which is both an
idempotent and a generator. For the punctured second order
Reed-Muller code $\mathcal{R}(2,m)^*,$ its idempotent is
    {
    \begin{equation*}
        \theta_0+\theta_1^*+\sum_{i=1}^{[\frac{m}{2}]}\theta_{l_i}^*, \quad
        l_j=1+2^j.
    \end{equation*}
    }

Before proving Theorem \ref{thm: m even sub-code family}, we
introduce the following two lemmata.
\begin{lem} \label{lem: gcd}
    {
    \begin{equation*}
        \gcd(2^m-1, 2^i+1)= \left\{\begin{array}{ll}
                                    1 & \quad \mbox{if } \gcd(m,2i)=\gcd(m,i),\\
                                2^{\gcd(m, i)}+1 & \quad \mbox{if } \gcd(m,2i)=2\gcd(m,i).
                                  \end{array}
                            \right.
    \end{equation*}
    }
\end{lem}

\begin{proof}
    {
    \begin{eqnarray*}
        \gcd(2^m-1, 2^i+1)  &\stackrel{(a)}{=}& \gcd(2^m-1, 2^i+1, 2^{m-i}-2^i) \\
                            &\stackrel{(b)}{=}& \gcd(2^m-1, 2^i+1, 2^{2i}-1) \\
                            &\stackrel{(c)}{=}& \gcd(2^{\gcd(m, 2i)}-1, 2^i+1),
    \end{eqnarray*}
    }
        where (a) is from the fact that
        $2^m-1=(2^{m-i}-1)(2^i+1)-(2^{m-i}-2^i);$ (b) is from the
        fact that $2^m-1=2^i(2^{m-i}-2^i)+(2^{2i}-1);$ (c) is from
        the fact that $\gcd(2^x-1, 2^y-1)=2^{\gcd(x,y)}-1.$

    If $\gcd(m,2i)=\gcd(m,i),$ we have
    {
        \begin{eqnarray*} 
        \gcd(2^m-1, 2^i+1)  &{=}& \gcd(2^{\gcd(m, 2i)}-1, 2^i+1)\\
                            &{=}& \gcd(2^{\gcd(m, i)}-1, 2^i+1)\\
                            &\stackrel{(d)}{\leq}& \gcd(2^{i}-1, 2^i+1)\\
                            &\stackrel{(e)}{=}& 1,
        \end{eqnarray*}
        }
        where (d) is from the fact that $2^{\gcd(m, i)}-1\mid
        2^i-1;$ (e) is from the fact that $\gcd(2^{i}-1,
        2^i+1)=1.$ Clearly we
        have in this case $\gcd(2^m-1, 2^i+1)=1.$

     If $\gcd(m,2i)=2\gcd(m,i),$ we have
    {
        \begin{eqnarray*} 
        \gcd(2^m-1, 2^i+1)  &{=}& \gcd(2^{\gcd(m, 2i)}-1, 2^i+1)\\
                            &{=}& \gcd(2^{2\gcd(m, i)}-1, 2^i+1)\\
                            &\stackrel{(f)}{=}& \gcd(2^{\gcd(m, i)}-1, 2^i+1)
                             \cdot \gcd(2^{\gcd(m, i)}+1, 2^i+1)\\
                            &{=}& \gcd(2^{\gcd(m, i)}+1, 2^i+1),\\
                            &\stackrel{(g)}{=}& 2^{\gcd(m, i)}+1,
        \end{eqnarray*}
        }
        where (f) is from the fact that $2^{2\gcd(m, i)}-1=(2^{\gcd(m, i)}-1)(2^{\gcd(m, i)}+1)$
        and $\gcd(2^{\gcd(m, i)}-1, 2^{\gcd(m, i)}+1)=1.$ (g) is from the fact that $2^{\gcd(m, i)}+1\mid
        2^i+1,$ since $\frac{i}{\gcd(m,i)}$ is odd due to $\gcd(m,2i)=2\gcd(m,i).$
    \end{proof}

\begin{lem}\label{lem: size of cyclotomic cosets}
    If $m=2t+1$ is odd, then for $l_i=1+2^i,$
        {
        $$|C_{l_i}|=m, \quad 1\leq i\leq  t.$$
        }
    If $m=2t+2$ is even, then  for $l_i=1+2^i,$
        {
        \begin{equation*}
            |C_{l_i}|=\left\{
                            \begin{array}{ll}
                                m &  1\leq i\leq  t,\\
                                m/2 & i=t+1.
                            \end{array}
                        \right.
        \end{equation*}
        }
\end{lem}

\begin{proof} If $m=2t+1$ is odd,  for $1\leq i \leq t,$ $l_i$ and $2^m-1$
are relatively prime according to Lemma \ref{lem: gcd}. In this
case, it is clear that $|C_{l_i}|=m.$

If $m=2t+2$ is even, for $1\leq i \leq t,$ $C_{l_i}$ consists of
    {
    $$\{1+2^i,2+2^{i+1}, 2^2+2^{i+2}, \cdots,2^t+2^{i+t}, 2^t+2^{i+t}, 2^{t+1}+2^{i+t+1} \cdots\}.$$}
It is easy to see that $|C_{l_i}|\geq t+2=m/2+1.$ Note that
$|C_{l_i}|$ must be either $m$ or a divisor of $m.$ Therefore,
$|C_{l_i}|=m$ for $1\leq i \leq t.$ Consider the case $i=t+1.$
$C_{l_{t+1}}$ consists of
$$\{1+2^{t+1},2+2^{t+2}, 2^2+2^{t+3}, \cdots,2^{t}+2^{2t+1}\}.$$ It
is easy to count that $|C_{l_{t+1}}|=t+1=m/2.$
\end{proof}

\begin{thm} \label{thm: m even sub-code family} Let $m=2t+2$ be
even, and let $d$ be any number in the range $1\leq d\leq t+1.$
Then there exists a
    {
    \begin{equation*}
        [2^m, m(t-d+2)+{m}/{2}+1, 2^{m-1}-2^{m-d-1}]
    \end{equation*}
    }
sub-code $\mathcal{R}_{2t+2}^{d}$ of $\mathcal{R}(2, m).$ It is
obtained by extending the cyclic sub-code of $\mathcal{R}(2,m)^*$
having idempotent
    {
    \begin{equation*}
        \theta_0+\theta_1^*+\sum_{j=d}^{t+1}\theta_{l_j}^*.
    \end{equation*}
    }
The code has codewords of weights $2^{m-1}$ and $2^{m-1}\pm
2^{m-h-1}$ for all $h$ in the range $d\leq h\leq t+1.$
\end{thm}

\begin{proof} Let $m=2t+2.$ The general codeword of $\mathcal{R}(2,m)^*$
is
    {
    \begin{equation*}
        b\theta_0+a_0 x^{i_{0}}\theta_1^*+\sum_{j=1}^{t+1}a_j x^{i_j}\theta_{l_j}^*, \quad
        l_j=1+2^j,
    \end{equation*}
    }
where $b, a_0, a_j\in GF(2),$ $0\leq i_k\leq 2^m-2$ for $0\leq
k\leq t$ and $0\leq i_{t+1}\leq 2^{m/2}-2.$  Consider the case
$b=0.$ Its Mattson-Solomon polynomial is
    {
    \begin{equation*}
        \sum_{s\in C_1}(\gamma_0 z)^s+\sum_{j=1}^{t+1}\sum_{s\in
        C_{l_j}}(\gamma_j z)^s,
    \end{equation*}
    }
where $\gamma_0, \gamma_j\in GF(2^m).$ Due to Theorem \ref{thm:
Mattson-Solomon polynomial coe}, the corresponding Boolean
function is
    {
    \begin{eqnarray*}
        S(\xi)&=&\sum_{j=0}^{t+1}T_{|C_{l_j}|}(\gamma_j
        \xi)^{l_j} \quad \mbox{for all } \xi\in GF(2^m)^*,\\
        &\stackrel{(a)}{=}& \sum_{j=0}^{t}T_{m}(\gamma_j
        \xi)^{1+2^j}+T_{m/2}(\gamma_{t+1} \xi)^{1+2^{t+1}}.
    \end{eqnarray*}
    }
Here (a) is due to Lemma \ref{lem: size of cyclotomic cosets}.
The corresponding symplectic form is %
    {
    \begin{eqnarray*} 
        \mathcal{B}(\xi,\eta)&=& S(\xi+\eta)+S(\xi)+S(\eta) \\
                            &\stackrel{(b)}{=}&\sum_{j=1}^{t}T_{m}(\gamma_j^{1+2^j}(\xi \eta^{2^j}+\xi^{2^j}\eta))
                            +T_{m/2}(\gamma_{t+1}^{1+2^{t+1}}(\xi \eta^{2^{t+1}}+\xi^{2^{t+1}}\eta))\\
                            &\stackrel{(c)}{=}& T_{m}\{\sum_{j=1}^{t}(\gamma_j^{1+2^j}\xi
                            \eta^{2^j}+\gamma_j^{1+2^{2t+2-j}}\xi\eta^{2^{2t+2-j}})
                            +\gamma_{t+1}^{1+2^{t+1}}\xi \eta^{2^{t+1}}\}\\
                            &=& T_{m}(\xi L_B(\eta)),
    \end{eqnarray*}
    }
where $\xi, \eta\in GF(2^m)^*$ and
    {
    \begin{equation*}
        L_B(\eta)=\sum_{j=1}^{t}\gamma_j[(\gamma_j \eta)^{2^j}+(\gamma_j \eta)^{2^{2t+2-j}}]
        +\gamma_{t+1}(\gamma_{t+1}\eta)^{2^{t+1}}.
    \end{equation*}
    }
Note that (b) is from the fact that
    $T_m(\alpha+\beta)=T_m(\alpha)+T_m(\beta)$ and
    {
    \begin{equation*}
        (\xi+\eta)^{1+2^j}=(\xi+\eta)(\xi^{2^j}+\eta^{2^j})=\xi^{1+2^j}+\xi
        \eta^{2^j}+\xi^{2^j} \eta+\eta^{1+2^j};
    \end{equation*}
    }
(c) is due to the fact that
    {
    \begin{eqnarray*} 
        T_{m}(\gamma_j^{1+2^j}\xi^{2^j}\eta)&=&T_{m}(\gamma_j^{1+2^{2t+2-j}}\xi\eta^{2^{2t+2-j}}),\\
        T_{m}(\gamma_{t+1}^{1+2^{t+1}}\xi \eta^{2^{t+1}})
            &=&T_{m/2}(\gamma_{t+1}^{1+2^{t+1}}(\xi \eta^{2^{t+1}}+\xi^{2^{t+1}}\eta)).\\
    \end{eqnarray*}
    }
Let $1 \leq d\leq t+1$ and
$\gamma_1=\gamma_2=\cdots=\gamma_{d-1}=0.$ Then
    {
    \begin{eqnarray*} 
        L_B(\eta)&=&\sum_{j=1}^{t}\gamma_j[(\gamma_j \eta)^{2^j}+(\gamma_j \eta)^{2^{2t+2-j}}]
                        +\gamma_{t+1}(\gamma_{t+1}\eta)^{2^{t+1}}\\
                 &=& \gamma_{d} (\gamma_{d} \eta)^{2^d}+\gamma_{d+1} (\gamma_{d+1}\eta)^{2^{d+1}}+\cdots
                 +\gamma_{t} (\gamma_{t}\eta)^{2^{t}}
                 +\gamma_{t+1}(\gamma_{t+1}\eta)^{2^{t+1}}\\
                 &&+\gamma_{t}(\gamma_{t} \eta)^{2^{t+2}}
                 +\cdots+\gamma_{d+1}(\gamma_{d+1} \eta)^{2^{2t+1-d}}+\gamma_{d} (\gamma_{d} \eta)^{2^{2t+2-d}}\\
                 &=& L'_B(\eta)^{2^d},
    \end{eqnarray*}
    }
where degree $L'_B(\eta)\leq 2^{2t+2-2d}.$ Thus the dimension of
the space of $\eta$ for which $L_B(\eta)=0$ is at most $2t+2-2d.$
So rank $B\geq 2t+2-(2t+2-2d)=2d$ (refer to (20) of Ch. 15 in
\cite{6}). In particular, when $d=t+1,$ rank $B=2t+2$ and the
symplectic matrix $B$ is corresponding to a quadratic bent
function.

Note that by setting $\gamma_i=0$ we are removing the idempotent
$\theta_{l_i}^*$ from the code. Setting the first $d-1$
$\gamma_i$'s equal to 0, we derive a sub-code
$\mathcal{R}_{2t+2}^{d'},$ which has a corresponding symplectic
form of rank $\geq 2d.$
Clearly the code $\mathcal{R}_{2t+2}^{d'}$ has idempotent
    {
    \begin{equation*}
        \theta_1^*+\sum_{j=d}^{t+1}\theta_{l_{j}}^*, \quad l_{j}=1+2^{j}.
    \end{equation*}
    }
According to Lemma \ref{lem: size of cyclotomic cosets}, the code
has dimension $\sum_{j=d}^{t+1}|C_{l_j}|=(t+2-d)m+m/2.$ Due to
Lemma \ref{lem: rank dis}, the code has codewords of weights
$2^{m-1}$ and $2^{m-1}\pm 2^{m-h-1}$ for all $h$ in the range
$d\leq h\leq t+1.$ Adding the all-one codeword into
$\mathcal{R}_{2t+2}^{d'},$ we get a sub-code
$\mathcal{R}_{2t+2}^{d^*}$ of $\mathcal{R}(2,m)^*$ having
idempotent
    {
    \begin{equation*}
        \theta_0+\theta_1^*+\sum_{j=d}^{t+1}\theta_{l_{j}}^*, \quad
        l_{j}=1+2^{j},
    \end{equation*}
    }
of dimension $(t+2-d)m+m/2+1$ and minimum distance
$2^{m-1}-2^{m-d-1}-1.$ Adding a parity check bit, we get the
extended code $\mathcal{R}_{2t+2}^{d}.$
\end{proof}

\begin{cor}
    {
    \begin{equation*}
        \mathcal{R}_{2t+2}^{t+1}\subset \mathcal{R}_{2t+2}^{t} \subset \cdots \subset \mathcal{R}_{2t+2}^1.
    \end{equation*}
    }
\end{cor}

\begin{proof}
Recall that $\mathcal{R}_{2t+2}^d$ for $1\leq d\leq t+1$ is by
extending the sub-code of $\mathcal{R}(2,m)^*,$
$\mathcal{R}_{2t+2}^{d^*},$ which has idempotent
    {
    \begin{equation*}
        \theta_0+\theta_1^*+\sum_{j=d}^{t+1}\theta_{l_{j}}^*, \quad l_{j}=1+2^{j}.
    \end{equation*}
    }
The corollary follows directly from the fact that
    {
    \begin{equation*}
      \mathcal{R}_{2t+2}^{{t+1}^*}\subset \mathcal{R}_{2t+2}^{t^*} \subset \cdots \subset
      \mathcal{R}_{2t+2}^{1^*}.
    \end{equation*}
    }\end{proof}

Clearly the sub-codes in the sub-code family for even $m$ satisfy
the nested structure. By a similar proof, the sub-codes in the
sub-code family for odd $m$ by Theorem \ref{thm: m odd sub-code
family} have the same property.
\begin{cor}
    {
    \begin{align*}
        &\mathcal{R}_{1,2t+1}^{t}\subset \mathcal{R}_{1,2t+1}^{t-1} \subset \cdots \subset
        \mathcal{R}_{1,2t+1}^1; \\
        &\mathcal{R}_{2,2t+1}^{t}\subset \mathcal{R}_{2,2t+1}^{t-1} \subset \cdots \subset \mathcal{R}_{2,2t+1}^1.
    \end{align*}
    }
\end{cor}

We say that a linear code is {\it minimum distance optimal} if it
achieves the largest minimum distance for given length and
dimension. If the binary code contains the all-one sequence, then
we say that it is {\it self-complementary}.

\begin{thm}\label{thm: even m dim 3m/2+1 distance optimal}
If a linear $[2^m, 1+3m/2, 2^{m-1}-2^{m/2-1}]$ code  is
self-complementary, then it is minimum distance optimal. Here $m$
is even.
\end{thm}

\begin{proof} Let $C$ be a self-complementary linear $[2^m, 1+3m/2,
2^{m-1}-2^{m/2-1}]$ code. If $C$ contains $\mathcal{R}(1,m),$ due
to Theorem 6 in Ch.14 in \cite{6}, it is clear that $C$ is minimum
distance optimal. Here we prove its is already true if $C$
contains the all-one codeword, i.e., $C$ is self-complementary.

Suppose that the minimum distance $d$ for given $n=2^m, k=1+3m/2$
can be larger, i.e., $d=2^{m-1}-2^{m/2-1}+\delta,$ where $\delta$
is a positive integer and $0<\delta<2^{m/2-1}.$ Due to the
Grey-Rankin bound (refer to (46) of Ch. 17 in \cite{6}) for the
binary self-complementary code, we have
    {
    \begin{eqnarray*}
        |C|&\leq& \frac{8d(n-d)}{n-(n-2d)^2}\\
            &=& \frac{8(2^{m-1}-2^{m/2-1}+\delta)(2^{m-1}+2^{m/2-1}-\delta)}{2^m-(2^{m/2}-2\delta)^2}\\
            &=& 2+\frac{2^{m-1}(2^m-1)}{\delta(2^{m/2}-\delta)}\\
            &\stackrel{(a)}{\leq}& 2^{3m/2-1}+2^{m-1}+2,
    \end{eqnarray*}
    }
where (a) follows from the fact that $\delta(2^{m/2}-\delta)\geq
2^{m/2}-1.$ However, $|C|=2^{1+3m/2}>2^{3m/2-1}+2^{m-1}+2,$ which
is impossible for even $m>0.$ Thus the self-complementary linear
$[2^m, 1+3m/2, 2^{m-1}-2^{m/2-1}]$ code has optimal minimum
distance.
\end{proof}
    {
         \begin{table}[ht]
            \caption{Some sub-codes of $\mathcal{R}(2,m)$} \label{table: subcodes m even}
            {\begin{center}
                            {
                \begin{tabular}{   l             | c | c | c | c | c | c    |c |c |c|c}
                        \hline
                                        $m$      & 4 & 6 & 6  & 8  & 8 & 8  &3 &5 &7 & 7\\
                                     length      &16 &64 &64  &256 &256&256 &8 &32&128&128\\
                                   dimension     & 7 &10 & 16 & 13 &21 & 29 &7 &11&15&22\\
                        \hline
                                     $d_{-}$ & 6 & 28 & 24 &120&112& 96 &2  &12&56& 48\\
                                minimum distance & 6 & 28 & 24 &120&112& 96 &2  &12&56& 48\\
                                     $d_{+}$ & 6 & 28 & 24 &122&116&111 &2  &12&56& 52\\
                        \hline
                \end{tabular}
                        }
            \end{center}
           }
        \end{table}
     }
{\it Remark:} In Table \ref{table: subcodes m even}, we show the minimum distance of the
sub-codes of $\mathcal{R}(2,m)$ constructed by Theorem \ref{thm: m even
sub-code family} and Theorem \ref{thm: m odd sub-code family} for $m\leq
8.$ Note that $d_{+}$ is the upper bound of the minimum distance for all the linear codes of the same length and dimension; $d_{-}$ is the
largest minimum distance, of which a linear code with the same length and dimension has been discovered so far (refer to
\cite{2}). From Table \ref{table: subcodes m even}, we see that although
all these sub-codes contain of $\mathcal{R}(1,m)$ and have specified
weight sets, they have minimum distances reaching the
upper bound $d_{+}$ or achieving the largest known minimum distance $d_{-}$ of which a linear code with the same length and dimension can be constructed (maybe not a sup-code of $\mathcal{R}(1,m)$ or a sub-code of $\mathcal{R}(2,m)$).
To some extend, we can say that the $\mathcal{R}(2,m)$ has good sub-codes that can be constructed by
Theorem \ref{thm: m even sub-code family} and Theorem \ref{thm: m odd
sub-code family}.

Compare the linear code $\mathcal{R}_{2t+2}^{t+1}$ to the Kerdock
code $\mathcal{K}(m).$  We consider $\mathcal{R}_{2t+2}^{t+1}$ as
$\mathcal{R}(1,m)$ together with $2^{m/2}-1$ cosets of
$\mathcal{R}(1,m),$ and $\mathcal{K}(m)$ as $\mathcal{R}(1,m)$
together with $2^{m-1}-1$ cosets of $\mathcal{R}(1,m).$ Note that
every coset is corresponding to a quadratic bent function and therefore
associated to a symplectic matrix of full rank. It is well known
that the cosets of $\mathcal{K}(m)$ are corresponding to the
maximal set of symplectic forms with the property that the rank of
the sum of any two in the set is still full rank. Clearly
$\mathcal{K}(m)$ has much more codewords. However,
$\mathcal{R}_{2t+2}^{t+1}$ enjoys a linear structure. One can
correspondingly obtain a set of $2^{m/2}-1$ symplectic matrices of
full rank, denoted as $\mathcal{G}^*.$ Introducing a matrix with
all zero elements into $\mathcal{G}^*,$ we get a set
$\mathcal{G}.$ Due to the linearity of $\mathcal{R}_{2t+2}^{t+1},$
$\mathcal{G}$ is a commutative group with respect to the addition
operation.
\begin{thm}
    For any even number $m,$ there exists a group $\mathcal{G}$ of $m\times m$ symplectic matrices
    with respect to the addition operation.
    There are $2^{m/2}$ symplectic matrices in $\mathcal{G}.$ In particular, all the
    matrices except the matrix with all zero elements have full
    rank $m.$
\end{thm}

\section{Conclusion} \label{sec: conclusion}
In this paper, we consider the first order and the second order
Reed-Muller codes. Our main contributions are twofold. First, we
prove the uniqueness of the first order Reed-Muller code.
Secondly, we give a linear sub-code family of the second order
Reed-Muller code $\mathcal{R}(2,m)$ for even $m,$ which is an
extension of Corollary 17 of Ch. 15 in \cite{6}. We also show that
for $m\leq 8,$ these specified sub-codes have good minimum
distance equal to the upper bound or the largest constructive minimum distance
for linear codes of the same length and dimension. As an
additional result, we obtain an additive commutative group of
$m\times m$ symplectic matrices of full rank with respect to the
addition operation, which is new to our knowledge.



\end{document}